\documentclass[aps,pre,preprint,a4paper]{revtex4-1}
\usepackage{amsthm}
\usepackage{amsmath}
\usepackage{amssymb}
\usepackage{graphicx}
\usepackage{dcolumn}
\usepackage{bm}

\newtheorem{thm}{Theorem}[]

\begin{document}


\title{Hamiltonian mechanics of generalized eikonal waves} 



\author{J. W. Burby}
 \affiliation{Princeton Plasma Physics Laboratory, Princeton, New Jersey 08543, USA}
\author{H. Qin}
 \affiliation{Princeton Plasma Physics Laboratory, Princeton, New Jersey 08543, USA}
 \affiliation{Dept. of Modern Physics, University of Science and Technology of China, Hefei, Anhui 230026, China}


\date{\today}

\begin{abstract}
In accordance with the Keller-Maslov global WKB theory, a semiclassical scalar wave field is best encoded as a triple consisting of (i) a Lagrangian submanifold $\Lambda$ in the ray phase space, (ii) a density $\mu$ on $\Lambda$, and (iii) an overall phase factor $\phi$. We present the Hamiltonian structure of the Cauchy problem for such a ``geometric semiclassical state" in the special case where the wave operator is Hermetian. Variational, symplectic, and Poisson formulations of the time evolution equations for $(\Lambda,\mu,\phi)$ are identitfied. Because we work in terms of the Keller-Maslov global WKB ansatz, as opposed to the more restrictive $\psi=a \exp(i S/\epsilon)$, all of our results are insensitive to the presence of caustics. In particular, because the variational principle is insensitive to caustics, the latter may be used to construct structure-perserving numerical integrators for scalar wave equations.
\end{abstract}


\maketitle 

\section{Introduction}
The eikonal assumption, $\psi=a\exp(i S/\epsilon)$, for scalar waves $\psi$ that obey the linear wave equation
\begin{align}\label{first}
\mathfrak{D}\psi=0,
\end{align} 
where $\mathfrak{D}$ is a Hermetian operator, and $\epsilon$ is the semiclassical ordering parameter, pervades much of the literature on semiclassical wave propagation. The reason for this is clear; eikonal waves share many properties with their idealized brethren, the plane waves. Unfortunately, eikonality is not a physically-robust property. As soon as caustics develop, the eikonal ansatz breaks down. In order to overcome this difficulty, while retaining many of the useful conceptual tools the eikonal assumption provides, J. B. Keller \cite{Keller_1958} and V. P. Maslov \cite{maslov} introduced a new class of semiclassical wave solutions for Eq.\,(\ref{first}). For lack of standard terminology, we will refer to these generalized eikonal waves as ``Keller-Maslov waves", or KMWs. Away from caustics, KMWs are superpositions of eikonal waves with distinct wave vectors. Near caustics, their momentum-space wavefunctions are eikonal. Detailed mathematical accounts of Keller and Maslov's theory (written in English) can be found in Refs.\,\cite{geo_asymp,Bates_97,Duistermaat_1974}. Also see Ref.\,\cite{Percival_1977}.

The purpose of this article is to describe the structure of the phase space for KMWs. We will identify the appropriate phase space as a set. Then we will show that the dynamics of a KMW in this phase space are Hamiltonian. In particular, we will give explicit expressions for the KMW Lagrange and Poisson brackets. Along the way, we will also identify a variational principle for KMWs that can be viewed as a generalization of the one studied in Ref.\,\cite{Kaufman_1987}. 

These results have notable practical and conceptual implications. The conceptual implications stem from the fact that caustic points and non-caustic points are on equal footing in the KMW phase space. In particular, any result that can be expressed in terms of this phase space is automatically insensitve to the presence of caustics. For example, by thinking along these lines, the theory of ponderomotive forces on waves developed in Ref. \cite{Dodin_2014} can be generalized to allow for caustics. On the practical side, our results should prove to be useful when numerically computing semiclassical solutions of wave equations. The Hamiltonian and variational structures we identify suggest how to develop structure-preserving integrators for this purpose. Because such integrators would be formulated in the KMW phase space, they would automatically handle caustic formation in a robust manner.

We will begin by identifying the KMW phase space as a set in Section\,\ref{what_is_phase_space}. In Section\,\ref{covariant}, we will present a variational formulation for KMW dynamics in this phase space. In Section\,\ref{cauchy_hamiltonian}, we will use this variational principle to express the KMW dynamical equations as a Hamiltonian system. Finally, in Section\,\ref{discussion} we will discuss our results. In particular, we will make a case for formulating semiclassical wave integrators in terms of the KMW phase space.


\section{What is the phase space for a Keller-Maslov wave?\label{what_is_phase_space}}
A physical system's phase space is its collection of allowable \emph{mechanical states}. A system's mechanical state at time $t$ consists of the minimal collection of temporally-local observables that uniquely determines the system's time evolution. For instance, by watching a planar pendulum swing, its spatial location and velocity vector can be measured at any instant of time $t_o$. According to Newton's second law, which a penduluum obeys with great accuracy, the pendulum's evolution for $t>t_o$ is completely determined by this information. Therefore, a model for the phase space of the pendulum is the set of positions and velocities that the pendulum could possibly achieve, which is the tangent bundle of the circle, $TS^1$. 

Mathematically, the phase space for a system described by a PDE on the spacetime, $M$, can be deduced using an approach invented in the context of gauge theories. First the \emph{covariant phase space}\,\cite{Lee_1990} (also see Ref.\,\cite{Khavkine_2014}) for the PDE is identified. This is a subset of the infinite-dimensional space of field configurations that consists of solutions of the PDE. Then a space-time decomposition, $M\approx Q\times\mathbb{R}$, is introduced and the covariant phase space is identified with the space of PDE initial data on $Q$. This space of initial data is the system's phase space. The remainder of this section is devoted to applying this methodology to KMWs.

KMWs are precisely defined in terms of a certain PDE on spacetime. Specifically, an \emph{abstract} or \emph{geometric} KMW consists of a triple $(\mathfrak{L},\mathfrak{m},\phi)$, where $\mathfrak{L}$ is a connected Lagrangian submanifold \cite{FoM,Weinstein_1971} $\mathfrak{L}\subset T^*M$ that satisfies the corrected Bohr-Sommerfeld condition, $\mathfrak{m}$ is a positive $N$-density with weight $1$ \cite{Paiva_2008} on $\mathfrak{L}$ ($N=\text{dim}(M)=\text{dim}(\mathfrak{L})$), and $\phi$ is a real number modulo $\epsilon\pi/2$. The physically-realizable abstract KMWs are those that satisfy the geometric PDE on spacetime,
\begin{align}
D|\mathfrak{L}=0\label{global_hamilton_jacobi_equation}\\
L_{X_D}\mathfrak{m}=0,\label{global_amplitude_transport_equation}
\end{align}
where $D:T^*M\rightarrow\mathbb{R}$ is the principal symbol of the operator $\mathfrak{D}$, $X_D$ is the Hamiltonian vector field on $T^*M$ with Hamiltonian function $D$, and $L_{X_{D}}$ denotes the Lie derivative along $X_D$. To see how this pair of equations arises directly from an asymptotic analysis of Eq.\,(\ref{first}) in the limit $\epsilon\rightarrow 0$, see Refs.\,\cite{geo_asymp,Bates_97}.
A concrete KMW $\psi$ is any wave field that can be recovered from an abstract KMW by applying Maslov's canonical map (see Refs.\,\cite{geo_asymp,Bates_97} for particularly lucid descriptions) $I(\mathfrak{L},\mathfrak{m},\phi)$ to the triple $(\mathfrak{L},\mathfrak{m},\phi)$\footnote{Strictly speaking}. Because the set of KMWs is identified with the set of abstract KMWs, it is permissible to work in terms of the abstract KMW instead of the concrete wave field.  From this perspective, Eqs.\,(\ref{global_hamilton_jacobi_equation}) and (\ref{global_amplitude_transport_equation}) can be viewed as the system of PDEs that govern KMW behavior. 

While a detailed description of how a wave field $\psi$ can be constructed from an abstract KMW is beyond the purview of this article, it is useful to briefly describe this construction in the absence of caustics. In the absence of caustics, a KMW is merely an eikonal wave, $\psi=a \exp(i S/\epsilon)$. The three constituents, $(\mathfrak{L},\mathfrak{m},\phi)$, of the abstract KMW that generate this wave via Maslov's canonical map are given as follows. The Lagrangian submanifold $\mathfrak{L}$ is the surface in the cotangent bundle of spacetime $T^*M$ swept out by the differential of the phase function, $\mathbf{d}S$. The density $\mathfrak{m}$ is essentially the squared wave amplitude $a^2$. Finally, the phase factor $\phi$ is the integration constant necessary to reconstruct the phase function $S$ from knowledge of the differential of $S$.  

Because Eqs.\,(\ref{global_hamilton_jacobi_equation}) and (\ref{global_amplitude_transport_equation}) comprise a PDE for KMWs, the covariant phase space for KMWs is the set of all $(\mathfrak{L},\mathfrak{m},\phi)$, where $\mathfrak{L}$ is a Lagrangian submanifold of $T^*M$ satisfying the corrected Bohr-Sommerfeld condition and Eq.\,(\ref{global_hamilton_jacobi_equation}), and $\mathfrak{m}$ is a positive $1$-density on $\mathfrak{L}$ satisfying Eq.\,(\ref{global_amplitude_transport_equation}). In short, the covariant phase space consists of all physically-realizable abstract KMWs. The phase space proper may therefore be identified by introducing an arbitrary space-time splitting $M\approx Q\times\mathbb{R}$ and parameterizing physically-realizable abstract KMWs by appropriate initial data on $Q$.

To this end, it is convenient to introduce a representation for our ``fields" $(\mathfrak{L},\mathfrak{m},\phi)$ in terms of generalized potentials. This approach parallels the procedure used in electrodynamics whereby the field strength $F$ is represented in terms of the $4$-potential $A$ as $F=\mathbf{d}A$. We set
\begin{align}
\mathfrak{L}&=\Phi(\mathcal{L}\times\mathbb{R})\\
\mathfrak{m}&=\Phi_{*}\mu|dt|\\
\phi&=\phi_{t_o}
\end{align}
where $\mathcal{L}$ is a connected $N$-dimensional manifold, $\mu|dt|$ is a positive density on $\mathcal{L}\times\mathbb{R}$, $\phi_t$ is a time-dependent real number modulo $\epsilon\pi/2$, and $\Phi:\mathcal{L}\times\mathbb{R}\rightarrow T^*(Q\times\mathbb{R})\approx T^*Q\times\mathbb{R}\times\mathbb{R}$ is a Lagrangian embedding of $\mathcal{L}\times\mathbb{R}$. We will assume $\Phi$ is of the form
\begin{align}
\Phi(x,t)=(\iota_t(x),t,K_t(x)),
\end{align}
where $\iota_t:\mathcal{L}\rightarrow T^*Q$ is a Lagrangian embedding of $\mathcal{L}$ into $T^*Q$ satisfying the corrected Bohr-Sommerfeld condition and 
\begin{align}
K_t(x)=\dot{\phi}_t+\int_{p}^x\left(\frac{\rm{d}}{\rm{d}t}\iota_t^*\vartheta\right)-\vartheta_{\iota_t(x)}(\dot{\iota}_t(x)).
\end{align}
The symbol $\vartheta$ denotes the canonical $1$-form on $T^*Q$ (in canonical coordinates $\vartheta=p_i\,\mathbf{d}q^i$) and $p\in\mathcal{L}$ is an arbitrary, but distinguished point in $\mathcal{L}$. All abstract KMWs that extend over the entire time axis can be expressed in this form. The time-dependent embedding $\iota_t:\mathcal{L}\rightarrow T^*Q$, the time-dependent number $\phi_t:\mathcal{L}\rightarrow\mathbb{R}~\text{mod}~\epsilon\pi/2$,  and the time-dependent density $\mu_t$ on $\mathcal{L}$ play the role of potentials for KMWs. 

Because these potentials admit the gauge transformation
\begin{align}
\iota_t&\mapsto \iota_t\circ\eta_t\\
\phi_t&\mapsto \phi_t + \int_{p}^{\eta_t(p)}\iota_t^*\vartheta \\
\mu_t&\mapsto \eta_t^*\mu_t,
\end{align}
where $\eta_t$ is an arbitrary time-dependent diffeomorphism of $\mathcal{L}$ that is isotopic to the identity, there are many possible physically-consistent dynamical laws they might satisfy. However, it is straightforward to show that, given a set of potentials for a physically-realizable KMW, there is always a gauge transformation that leads to potentials satisfying the equations
\begin{align}
&\dot{\iota}_t(x)=X_{E_t}(\iota_t(x))\label{idot}\\
&\dot{\phi}_t=\left(\vartheta(X_{E_t})-E_t\right)(\iota_t(p))\label{phidot}\\
&\frac{\rm{d}}{\rm{d}t}\left(\iota_t^*\left(\rho_t\right)\mu_t\right)=0,\label{mudot}
\end{align} 
where $E_t$ is a time-dependent function on $T^*Q$ that is defined implicitly by the relation $D(z,t,-E_t(z))=0$, $X_{E_t}$ is the Hamiltonian vector field on $T^*Q$ with Hamiltonian $E_t$, and $\rho_t(z)=\partial D/\partial U(z,t,-E_t(z))$. Note that this definition of $E_t$ requires $\partial D/\partial U(z,t,U)\neq 0$ along $\mathfrak{L}$ in $T^*M$.

Equations\,(\ref{idot}), (\ref{phidot}), and (\ref{mudot}) prescribe an initial-value problem on $Q$ that can be used to generate KMWs. For conveniece, we will refer to them as the KMW potential equations of motion, or the KMWP equations of motion. This is not an entirely-physical initial-value problem, however. Initial conditions that are related by time-independent gauge transformations generate the same abstract KMW. In other words, at a given time $t$, the triple $(\iota_t,\phi_t,\mu_t)$ is not a physically-meaningful quantity. Nevertheless, \emph{the orbit} of $(\iota_t,\phi_t,\mu_t)$ under the action of the time-independent gauge transformations, which we will denote $[(\iota_t,\phi_t,\mu_t)]$, is a physical quantity. The evolution of an orbit determines the evolution of a KMW, and vice versa. Therefore, the phase space for KMWs, $\mathcal{P}$, is the quotient of $(\iota,\phi,\mu)$-space by time-independent gauge transformations. It can also be shown that KMWP equations of motion uniquely determine a dynamical system on the KMW phase space. This follows from the symmetry of Eqs.\,(\ref{idot}), (\ref{phidot}), and (\ref{mudot}) under time-independent gauge transformations.

We summarize this result with the following theorem. Let $\mathcal{L}$ be an $N=\text{dim}(Q)$-dimensional manifold with distinguished point $p\in\mathcal{L}$ and $\text{den}(\mathcal{L})$ the set of positive densities on $\mathcal{L}$. 
Fix an integral de Rham class $a\in H^1(\mathcal{L},\mathbb{R})$ and let $I$ be a connected component of the collection of Lagrangian embeddings $\iota:\mathcal{L}\rightarrow T^*Q$ that satisfy
\begin{align}
\frac{\epsilon\pi}{2}\mu_\iota+[\iota^*\vartheta]=2\pi\epsilon\,a,
\end{align}
where $\mu_\iota$ is the Maslov class of the embedding $\iota$ \cite{Arnold_1967}.
Set $\mathcal{P}_o=I\times\left(\mathbb{R}~\text{mod}\frac{\epsilon\pi}{2}\right)\times\text{den}(\mathcal{L})$.
\begin{thm} \label{thm0}
There is a free left $\text{Diff}(\mathcal{L})_o$-action on the space of KMW potentials, $\Phi:\text{Diff}(\mathcal{L})_o\times \mathcal{P}_o\rightarrow\mathcal{P}_o $, given by
\begin{align}
\Phi_\eta(\iota,\phi,\mu)=\left(\iota\circ\eta^{-1},\phi+\int_p^{\eta^{-1}(p)}\iota^*\vartheta,\eta_*\mu\right).
\end{align}
Here $\text{Diff}(\mathcal{L})_o$ is the group of diffeomorphisms of $\mathcal{L}$ that are isotopic to the identity. The phase space for KMWs is $\mathcal{P}=\mathcal{P}_o/\text{Diff}(\mathcal{L})_o$.
\end{thm}

Notice that the phase space for KMWs is the base of an infinite-dimensional principal-$\text{Diff}(\mathcal{L})_o$ bundle. This allows us to replace statements pertaining to the phase space $\mathcal{P}$ with gauge-invariant statements pertaining to the space of potentials $\mathcal{P}_o$. Because calculations are somewhat less cumbersome when performed in $\mathcal{P}_o$, we will adopt this practice in what follows.

The quotient $\mathcal{P}_o/\text{Diff}(\mathcal{L})_o$ can be described in concrete terms. Suppose we are given a point $(\iota,\phi,\mu)\in\mathcal{P}_o$. The following gauge-invariant geometric data can be extracted from this point. The image of $\mathcal{L}$ under $\iota$, $\Lambda\equiv\iota(\mathcal{L})$, is gauge invariant because the effect of a gauge transformation on $\iota$ amounts to a relabeling of $\mathcal{L}$. For the same reason, the pushforward $m=\iota_*\mu$ is a gauge-invariant positive density on $\Lambda$. This ``quantum relabeling symmetry" plays a similar role in KMW theory as particle relabeling symmetry plays in Euler-Poincar\'e theory \cite{Holm_1998}. Let $E=\Lambda\times\mathbb{R}~\text{mod}\frac{\epsilon\pi}{2}$ be the trivial $\mathbb{R}~\text{mod}\frac{\epsilon\pi}{2}$-bundle over $\Lambda$. Equip $E$ with the principal connection $\mathcal{A}=\mathbf{d}\phi-\iota_\Lambda^*\vartheta$, where $\phi$ is the $\mathbb{R}~\text{mod}\frac{\epsilon\pi}{2}$ coordinate on $E$ and $\iota_\Lambda$ is the canonical inclusion $\Lambda\rightarrow T^*Q$. The phase factor $\phi$ determines a gauge-invariant parallel section, $s$, of $E$. Explicitly, for $l\in\Lambda$, $s(l)=(l,\tilde{\phi}(l))$, where
\begin{align}
\tilde{\phi}(l)=\phi+\int_{\iota(p)}^l\iota_\Lambda^*\vartheta.
\end{align}
The gauge-invariant triple $(\Lambda,\tilde{\phi},m)$ can be seen to completely characterize the orbit $[(\iota,\phi,\mu)]$. Thus $\mathcal{P}_o/\text{Diff}(\mathcal{L})_o$ consists of Lagrangian submanifolds of $T^*Q$ equipped with a positive density and a parallel section of the $\mathbb{R}~\text{mod}\frac{\epsilon\pi}{2}$-bundle $E$ \footnote{The bundle $E$ is essentially the $\mathbb{Z}_4$-reduction of the \emph{phase bundle} described in Ref.\,\cite{Bates_97}. A completely rigorous treatment of KMWs would require replacing $E$ with the phase bundle. The impact of this fact on our analysis is that the overall phases of our KMWs are only determined mod $\mathbb{Z}_4$. This minor deficiency is not difficult to repair, but doing so would significantly complicate our presentation.}.

If we were to modify the definition of $\mathcal{P}$ by dropping the Maslov class from the corrected Bohr-Sommerfeld condition and allowing for multi-valued $\tilde{\phi}$, the resulting space would be the cotangent bundle of the collection of \emph{isodrastic Planckian manifolds} \cite{Weinstein_1990}. The KMW phase space, $\mathcal{P}$, is therefore the cotangent bundle of the collection of \emph{quantizable} isodrastic Planckian manifolds. Here, ``quantizable" refers to imposing the corrected Bohr-Sommerfeld condition. In Ref.\,\cite{Weinstein_1990}, Weinstein studied isodrastic Planckian manifolds in order to develop a classical analogue of Berry's phase\,\cite{Berry_1984}. In fact, Weinstein introduced these Planckian manifolds with a good understanding of their relationship with semiclassical scalar waves. It is therefore interesting that Ref.\,\cite{Weinstein_1990} does not explicitly relate the space of isodrastic planckian manifolds to the phase space for KMWs. In particular, it is interesting that this work contains no discussion of the relationship between the Weinstein symplectic form (see Eq.\,(2) in Ref.\cite{Weinstein_1990}) on the set of \emph{isodrastic weighted Lagrangian submanifolds} and a Hamiltonian formulation of semiclassical scalar wave dynamics. As we will see, the Weinstein symplectic form is closely related to the symplectic form on the KMW phase space in the special case where $\mathfrak{D}$ generates the Schr{\"o}dinger equation.


\section{Variational principle for KMWs\label{covariant}}
We now turn to the task of demonstrating that the KMW phase space is actually an infinite-dimensional symplectic manifold and that KMW dynamics are Hamiltonian relative to this symplectic structure. We will exhibit the symplectic structure relevant to general KMW dynamics by developing and then manipulating an appropriate variational principle. This section is devoted to the development of the variational principle.

Our method for identifying this variational principle has its roots in Ref.\,\cite{Kaufman_1987} , which derived a spacetime-covariant variational principle for eikonal vector-valued waves. Starting from a standard (see Refs.\,\cite{Dodin1_2014,brizard_rays}) spacetime-covariant variational principle for Eq.\,(\ref{first}), we will implement the Keller-Maslov ansatz by treating it as an (asymptotic) holonomic constraint. This will lead directly to a spacetime-covariant variational principle for KMWs. We will then break the manifest space-time covariance in order to finally arrive at a variational principle for the KMW time evolution laws.

As is readily verified, solutions of Eq.\,(\ref{first}) are critical points of the action functional
\begin{align}
\mathfrak{S}(\psi)=\text{Re}\int_{M}\psi_{\mathcal{R}}^*(m)\mathfrak{D}\psi_{\mathcal{R}}(m)\,dm, 
\end{align}
where $\mathcal{R}\subset M$ is an arbitrary compact subset of the spacetime, $\psi_{\mathcal{R}}(m)=\psi(m)$ when $m\in\mathcal{R}$, $\psi_{\mathcal{R}}(m)=0$ when $m\not\in\mathcal{R}$, and $dm$ is the measure on spacetime relative to which $\mathfrak{D}$ is self-adjoint, i.e.
\begin{align}
&\varphi\in L_2(M)\Rightarrow \nonumber\\
&\int\varphi^*(m)\mathfrak{D}\varphi(m)\,dm=\int (\mathfrak{D}\varphi(m))^*\varphi(m)\,dm.
\end{align} 
The localized wavefunction $\psi_{\mathcal{R}}$ appears in the action instead of $\psi$ itself in order to ensure that the action integral converges. When varying the action, it is required that $\psi$ is held fixed on the boundary of $\mathcal{R}$. Note that when $\mathfrak{D}=i\hbar\frac{\partial}{\partial t}-\hat{H}_t$ is the wave equation operator for the Schr\"odinger equation, this variational principle is essentially the Dirac-Frenkel principle \cite{Frenkel_1934}.

Now suppose we restrict our attention to solutions of the wave equation that are asymptotically KMWs. With the term ``asymptotic", we refer to the typical semiclassical limit, where the ratio, $\epsilon$, of the wave spatio-temporal scale to that of the background tends to zero. These KMWs should be critical points of the action function $\mathfrak{S}$ in the limit $\epsilon\rightarrow 0$. Therefore, if we constrain the $\psi$ in $\mathfrak{S}(\psi)$ to be a KMW and then calculate the leading-order contribution to the action, we should expect that the result will serve as a variational principle for KMWs. 

Computing the leading-order contribution to the action involves a number of somewhat-subtle applications of the method of stationary phase (Refs.\,\cite{Hormander_1971,Duistermaat_1974} are helpful in this regard). The result is simply
\begin{align}
\mathcal{S}(\mathfrak{L},\mathfrak{m},\phi)&\equiv \lim_{\epsilon\rightarrow 0}\mathfrak{S}(I(\mathfrak{L},\mathfrak{m},\phi))\nonumber\\
&=\int_{\mathfrak{L}}D\,\mathfrak{m}_{\mathfrak{R}},
\end{align}
where $\mathfrak{R}\subset\mathfrak{L}$ is an arbitrary compact subset of $\mathfrak{L}$, $\mathfrak{m}_{\mathfrak{R}}(l)=\mathfrak{m}(l)$ if $l\in \mathcal{R}$, and $\mathfrak{m}(l)=0$ if $l\not\in\mathfrak{R}$. It is straightforward to verify that critical points of $\mathcal{S}$ are precisely the physically-realizeable KMWs. Thus, we have indeed identified a spacetime-covariant variational principle for KMWs, as we expected.
 
This spacetime variational principle can now be used to construct a variational principle for the time evolution laws satisfied by KMWs. First, we introduce an arbitrary space-time splitting $M\approx Q\times\mathbb{R}$ and express the abstract KMW $(\mathfrak{L},\mathfrak{m})$ in terms of the potentials $(\iota_t,\phi_t,\mu_t)$ introduced in the previous section. We then set $\mathfrak{R}=\Phi(\mathcal{L}\times[t_1,t_2])$, where $\Phi:\mathcal{L}\times\mathbb{R}\rightarrow T^*M$ is constructed in terms of the potentials as in the previous section. Note that, because $\mathcal{L}$ need not be compact, this choice of $\mathfrak{R}$ is not necessarily compact. If $\mathcal{L}$ is not compact, then we add the condition that $\mu_t$ should be integrable, i.e. $\int_{\mathcal{L}}\mu_t<\infty$. With these choices in place, the action for KMWs can now be expressed as
\begin{align}
&\mathcal{S}_{[t_1,t_2]}(\iota_\cdot,\phi_\cdot,\mu_\cdot)=\nonumber\\
&\int_{t_1}^{t_2}\int_{\mathcal{L}}D(\iota_t(x),t,K_t(x))\,\mu_t(x)\,dt,
\end{align}
where, as in the previous section
\begin{align}
K_t(x)=\dot{\phi}_t+\int_{p}^x\left(\frac{\rm{d}}{\rm{d}t}\iota_t^*\vartheta\right)-\vartheta_{\iota_t(x)}(\dot{\iota}_t(x)).
\end{align}
Note that we are regarding $\mathcal{S}_{[t_1,t_2]}$ as a functional of paths in $(\iota,\phi,\mu)$-space with fixed endpoints.
 
 
While it can be inferred from the analysis presented so far, we will now prove directly that this action functional has as critical points those curves $t\mapsto (\iota_t,\phi_t,\mu_t)$ that satisfy the dynamical equations for the KMW potentials, Eqs.\,(\ref{idot}), (\ref{phidot}), and (\ref{mudot}).
\begin{thm} \label{thm1}
Let $\iota_t:\mathcal{L}\rightarrow T^*Q$, $\phi_t$, and $\mu_t$, be a time-dependent Lagrangian embedding that satisfies the corrected Bohr-Sommerfeld condition, a time-dependent real number modulo $\epsilon\pi/2$, and a time-dependent positive $N$-density with weight $1$ on $\mathcal{L}$, respectively. If these quantities satisfy
\begin{align}
&\dot{\iota}_t(x)=X_{E_t}(\iota_t(x))\\
&\dot{\phi}_t=\left(\vartheta(X_{E_t})-E_t\right)(\iota_t(p))\\
&\frac{\rm{d}}{\rm{d}t}\left(\iota_t^*\left(\frac{\partial D}{\partial U}\right)\mu_t\right)=0,
\end{align}
then the curve $t\mapsto (\iota_t,\phi_t,\mu_t)$ is a critical point of the action
\begin{align}
&\mathcal{S}_{[t_1,t_2]}(\iota_\cdot,\phi_\cdot,\mu_\cdot)=\nonumber\\
&\int_{t_1}^{t_2}\int_{\mathcal{L}}D(\iota_t(x),t,K_t(x))\,\mu_t(x)\,dt,
\end{align}
regarded as a functional of curves in $(\iota,\phi,\mu)$-space with fixed endpoints.
\end{thm}
\begin{proof}
We will merely compute an arbitrary variation of $\mathcal{S}_{[t_1,t_2]}$ and then verify that this variation is zero at curves in $(\iota,\phi,\mu)$-space of the specified form. Let $(t,\epsilon)\mapsto(\iota_t^\epsilon,\phi_t^\epsilon,\mu_t^\epsilon)$ be an arbitrary two-parameter curve in $(\iota,\phi,\mu)$-space with fixed endpoints and set $(\iota_t,\phi_t,\mu_t)=(\iota_t^0,\phi_t^0,\mu_t^0)$. The variations $\delta\phi_t=\frac{\rm{d}}{\rm{d}\epsilon}\big|_0 \phi_t^\epsilon$ and $\delta\mu_t=\frac{\rm{d}}{\rm{d}\epsilon}\big|_0\mu_t^\epsilon$ are arbitrary real numbers and densities, respectively. On the other hand, $\delta\iota_t=\frac{\rm{d}}{\rm{d}\epsilon}\big|_0\iota_t^\epsilon$, which is a section of the pullback bundle $\iota_t^*T(T^*Q)$, is constrained by the corrected Bohr-Sommerfeld condition. The corrected Bohr-Sommerfeld condition \cite{geo_asymp,Bates_97} requires that a certain linear combination of the Maslov class of $\iota_t^\epsilon$, $m_t^\epsilon\in H^1(\mathcal{L},\mathbb{R})$, and the Lagrange class $[\iota_t^{\epsilon*}\vartheta]\in H^1(\mathcal{L},\mathbb{R})$ be integral. Because the Maslov class is neccessarily parameter-independent for any parameter-dependent Lagrangian embedding, this condition is equivalent to the requirement that $\iota_t^\epsilon$ always remain within a distinguished connected component of a level set of the mapping $\iota\mapsto[\iota^*\vartheta]$. Therefore, the de Rham class $[\iota_t^\epsilon\vartheta]\in H^1(\mathcal{L})$ must satisfy
\begin{align}
\frac{\rm{d}}{\rm{d}\epsilon}\bigg|_0[\iota_t^\epsilon\vartheta]=0.
\end{align} 
By de Rham's theorem, this last condition is equivalent to $\frac{\rm{d}}{\rm{d}\epsilon}\big|_0\iota_t^{\epsilon*}\vartheta$ being exact. For the same reasons, $\frac{\rm{d}}{\rm{d}t}\big|\iota_t^{*}\vartheta$ must also be exact. The exactness of this pair of differential forms then implies that there must be functions $\delta g_t,h_t:T^*Q\rightarrow\mathbb{R}$ such that
\begin{align}
\delta\iota_t(x)=X_{\delta g_t}(\iota_t(x))\\
\dot{\iota}_t(x)=X_{h_t}(\iota_t(x)).
\end{align}
With this in mind, it is readily seen that the first variation of the action is given by
\begin{align}\label{first_variation}
&\frac{\rm{d}}{\rm{d}\epsilon}\bigg|_0\mathcal{S}_{[t_1,t_2]}(\iota_\cdot^\epsilon,\phi_\cdot^\epsilon,\mu_\cdot^\epsilon)=\\
&\int_{t_1}^{t_2}\int_{\mathcal{L}}\bigg(\iota_t^*\left\{E_t-h_t,\delta g_t\right\}\bigg)\frac{\partial D}{\partial U}\,\mu_t\,dt\nonumber\\
+&\int_{t_1}^{t_2}\int_{\mathcal{L}}\iota_t^*\delta g_t\bigg(\frac{\rm{d}}{\rm{d}t}\bigg(\frac{\partial D}{\partial U}\,\mu_t\bigg)\bigg)\,dt\nonumber\\
+&\int_{t_1}^{t_2}\int_{\mathcal{L}}\bigg(\vartheta_{\iota_t(p)}(\delta\iota_t(p))-\delta g_t(\iota_t(p))-\delta\phi_t\bigg)\frac{\rm{d}}{\rm{d}t}\bigg(\frac{\partial D}{\partial U}\,\mu_t\bigg)\,dt\nonumber\\
+&\int_{t_1}^{t_2}\int_{\mathcal{L}}D\,\delta\mu_t\,dt,\nonumber
\end{align}
where $\partial D/\partial U(x,t)=\frac{\partial D}{\partial U}(\iota_t(x),t,K_t(x))$ and $D(x,t)=D(\iota_t(x),t,K_t(x))$.

 If $(\iota_t,\phi_t,\mu_t)$ satisfies the desired equations, then it is simple to verify that $h_t=E_t+\text{const}$, and $K_t=-\iota_t^*E_t$. Under these conditions, the right-hand-side of Eq.\,(\ref{first_variation}) vanishes.
\end{proof}
\section{KMW dynamics as a Hamiltonian system\label{cauchy_hamiltonian}}
With the variational principle in the previous section, we have formulated KMW dynamics as a gauge theory. We can therefore draw upon the general result \cite{Peierls_1952} (see also Refs.\,\cite{Marsden_1998,Khavkine_2014} for more recent discussions of the idea) that gauge-theoretic dynamical equations always possess a Hamiltonian structure.  In this section, we will apply this result in order to explicitly identify the Hamiltonian structure underlying KMW dynamics. 

Recall that $\mathcal{P}_o$ denotes $(\iota,\phi,\mu)$-space. Associated with the KMWP equations of motion on $\mathcal{P}_o$, is the time-advance map $\mathcal{F}_{t,t_o}:\mathcal{P}_o\rightarrow\mathcal{P}_o$. By definition, if $(\iota_t,\phi_t,\mu_t)$ is a solution of the KMWP equations of motion, then $\mathcal{F}_{t,t_o}(\iota_{t_o},\phi_{t_o},\mu_{t_o})=(\iota_t,\phi_t,\mu_{t})$. This mapping satisfies the so-called time-dependent flow property \cite{FoM},
\begin{align}
\mathcal{F}_{t_3,t_2}\circ\mathcal{F}_{t_2,t_1}=\mathcal{F}_{t_3,t_1}.
\end{align}
It can also be used to define a mapping from $\mathcal{P}_o$, regarded as a space of initial conditions for the KMWP equations of motion, into the space of paths in $\mathcal{P}_o$; an initial condition is mapped into its corresponding solution of Eqs.\,(\ref{idot}), (\ref{phidot}), and (\ref{mudot}). If $\mathcal{C}(\mathcal{P}_o)$ is the set of all conceivable paths $t\mapsto (\iota_t,\phi_t,\mu_t)$, then we denote this mapping $\mathcal{F}_{t_o}:\mathcal{P}_o\rightarrow\mathcal{C}(\mathcal{P}_o)$ (note that there is only one subscript, and therefore no possibility of confusion with $\mathcal{F}_{t,t_o}$). By definition, $\mathcal{F}_{t_o}(\iota,\phi,\mu)$ is the path given by
\begin{align}
\mathcal{F}_{t_o}(\iota,\phi,\mu)(t)=\mathcal{F}_{t,t_o}(\iota,\phi,\mu).
\end{align}

By composing the map $\mathcal{F}_{t_o}$ with the action $\mathcal{S}_{[t_o,t]}$ from the previous section, we obtain the restricted action $\hat{\mathcal{S}}_{t,t_o}:\mathcal{P}_o\rightarrow\mathbb{R}$,
\begin{align}
\hat{\mathcal{S}}_{t,t_o}(\iota,\phi,\mu)\equiv\mathcal{S}_{[t_o,t]}\circ\mathcal{F}_{t_o}(\iota,\phi,\mu).
\end{align}
Let us calculate the first exterior derivative of the restricted action. On the one hand, the result must be zero because $D=0$ along a solution of Eqs.\,(\ref{idot}), (\ref{phidot}), and (\ref{mudot}). On the other hand, the restricted action is given as a composition of mappings. Therefore, the exterior derivative can be computed using the chain rule. Indeed, if $(t,\epsilon)\mapsto(\iota_t^\epsilon,\phi_t^\epsilon,\mu_t^\epsilon)$ is an arbitrary two-parameter curve in $\mathcal{P}_o$ (free endpoints), then
\begin{align}\label{free_endpoints}
&\frac{\rm{d}}{\rm{d}\epsilon}\bigg|_0\mathcal{S}_{[t_o,t]}(\iota^\epsilon_\cdot,\phi^\epsilon_\cdot,\mu^\epsilon_\cdot)=\\
&\int_{t_0}^{t}\int_{\mathcal{L}}\bigg(\iota_t^*\left\{E_t-h_t,\delta g_t\right\}\bigg)\frac{\partial D}{\partial U}\,\mu_t\,dt\nonumber\\
+&\int_{t_o}^{t}\int_{\mathcal{L}}\iota_t^*\delta g_t\bigg(\frac{\rm{d}}{\rm{d}t}\bigg(\frac{\partial D}{\partial U}\,\mu_t\bigg)\bigg)\,dt\nonumber\\
+&\int_{t_o}^{t}\int_{\mathcal{L}}\bigg(\vartheta_{\iota_t(p)}(\delta\iota_t(p))-\delta g_t(\iota_t(p))-\delta\phi_t\bigg)\frac{\rm{d}}{\rm{d}t}\bigg(\frac{\partial D}{\partial U}\,\mu_t\bigg)\,dt\nonumber\\
+&\int_{t_o}^{t}\int_{\mathcal{L}}D\,\delta\mu_t\,dt\nonumber\\
+&\int_{\mathcal{L}}\bigg(\delta\phi_t+\delta g_t(\iota_t(p))-\vartheta_{\iota_t(p)}(\delta\iota_t(p))-\iota_t^*\delta g_t\bigg)\frac{\partial D}{\partial U}\mu_t\bigg|^t_{t_o}.\nonumber
\end{align}
Here, as in the previous section, we have
\begin{align}
\delta\iota_t(x)=X_{\delta g_t}(\iota_t(x))\\
\dot{\iota}_t(x)=X_{h_t}(\iota_t(x)).
\end{align}
If $(\iota_t^\epsilon,\phi_t^\epsilon,\mu_t^\epsilon)$ is a solution of KMWP equations of motion, then only the boundary terms on the right-hand-side of Eq.\,(\ref{free_endpoints}) are nonzero. Therefore, if $(\delta\iota,\delta\phi,\delta\mu)$ is a tangent vector at $(\iota,\phi,\mu)$, the chain rule gives
\begin{align}\label{one_form_conservation}
\mathbf{d}\hat{\mathcal{S}}_{t,t_o}(\delta\iota,\delta\phi,\delta\mu)&=\frac{\rm{d}}{\rm{d}\epsilon}\bigg|_0\hat{\mathcal{S}}_{t,t_o}(\iota^\epsilon,\phi^\epsilon,\mu^\epsilon)\\
&=\frac{\rm{d}}{\rm{d}\epsilon}\bigg|_0\mathcal{S}_{[t_o,t]}(\mathcal{F}_{t,t_o}(\iota^\epsilon,\phi^\epsilon,\mu^\epsilon))\nonumber\\
&=(\mathcal{F}_{t,t_o}^*\Theta_t-\Theta_{t_o})(\delta\iota,\delta\phi,\delta\mu),\nonumber
\end{align}
where $\Theta_t$ is the time-dependent one-form on $\mathcal{P}_o$ given by
\begin{align}\label{pre_one_form}
\Theta_t(\delta\iota,\delta\phi,\delta\mu)=&-\int_{\mathcal{L}}\iota^*(\delta g\,\rho_t)\,\mu\\
-&p_{\phi,t}(\iota,\mu)\left(\vartheta_{\iota(p)}(\delta\iota(p))-\delta g(\iota(p))-\delta\phi\right);\nonumber
\end{align}
$\rho_t:T^*Q\rightarrow\mathbb{R}$ is given by $\rho_t(z)=\frac{\partial D}{\partial U}(z,t,-E_t(z))$; and 
\begin{align}
p_{\phi,t}(\iota,\mu)=\int_{\mathcal{L}}(\iota^*\rho_t)\,\mu.
\end{align}

Because $\hat{\mathcal{S}}_{t,t_o}=0$, Eq.\,(\ref{one_form_conservation}) implies that the one-form $\Theta_t$ is ``frozen-in" to the flow on $\mathcal{P}_o$,
\begin{align}
\mathcal{F}_{t,t_o}^*\Theta_t=\Theta_{t_o}.
\end{align}
This infinite-dimensional frozen-in law is very nearly a statement of the fact that the dynamical equations for the potentials are a Hamiltonian system. Indeed, by taking the exterior derivative of the frozen-in law, we obtain the conservation law
\begin{align}\label{proto_symplecticity}
\mathcal{F}_{t,t_o}^*\mathbf{d}\Theta_t=\mathbf{d}\Theta_{t_o}.
\end{align}
Moreover, by taking the time derivative of the frozen-in law, we obtain an equation reminiscent of the time-dependent Hamilton's equations,
\begin{align}\label{proto_hamilton}
\text{i}_{X_t}\mathbf{d}\Theta_t=-\mathbf{d}\mathcal{E}_t-\dot{\Theta}_t,
\end{align}
where $X_t$ is the vector field on $\mathcal{P}_o$ defined by Eqs.\,(\ref{idot}), (\ref{phidot}), and (\ref{mudot}), and $\mathcal{E}_t:\mathcal{P}_o\rightarrow\mathbb{R}$ is the energy functional
\begin{align}
\mathcal{E}_t(\iota,\phi,\mu)=\Theta_t(X_t)(\iota,\phi,\mu)=-\int_{\mathcal{L}}\iota^*(E_t\,\rho_t)\mu.
\end{align}

In fact, the two-form $-\mathbf{d}\Theta_t$ is not a symplectic form because it is degenerate. Therefore Eq.\,(\ref{proto_hamilton}) is not an instance of Hamilton's equations. However, the tangent vectors that annihiliate $-\mathbf{d}\Theta_t$ are precisely the infinitesimal generators of time-independent gauge transformations. Moreover, $-\mathbf{d}\Theta_t$ is preserved by time-independent gauge transformations. This pair of properties is shared by $\mathbf{d}\mathcal{E}_t$, $\Theta_t$, and $\dot{\Theta}_t$ as well. It follows that the two-form $-\mathbf{d}\Theta_t$, the one-forms $\Theta_t,\dot{\Theta}_t$, and the energy functional $\mathcal{E}_t$ induce corresponding objects on the quotient space $\mathcal{P}=\mathcal{P}_o/\text{Diff}_o(\mathcal{L})$. Let $\Pi:\mathcal{P}_o\rightarrow\mathcal{P}$ be the projection map that sends the potentials $(\iota,\phi,\mu)$ to the corresponding orbit $[(\iota,\phi,\mu)]$. There is a unique one-form on $\mathcal{P}$, $\theta_t$, that satisfies
\begin{align}\label{def_reduced}
\Pi^*\theta_t&=\Theta_t\\
\Pi^*\dot{\theta}_t&=\dot{\Theta}_t.
\end{align}
Likewise, there is a unique functional on $\mathcal{P}$, $e_t$, that satisfies
\begin{align}
\Pi^*e_t=\mathcal{E}_t.
\end{align}
By Eq.\,(\ref{proto_hamilton}), these quantities are related by
\begin{align}\label{hamilton}
\text{i}_{x_t}\mathbf{d}\theta_t=-\mathbf{d}e_t-\dot{\theta}_t,
\end{align}
where $x_t$ is the dynamical vector field on the KMW phase induced by KMWP equations of motion. With Eq.\,(\ref{hamilton}) we have formally succeeded in realizing KMW dynamics as a time-dependent Hamiltonian system. 

Explicit expressions for the presymplectic form $-\mathbf{d}\Theta_t$, the one-form $\theta_t$, the symplectic form $-\mathbf{d}\theta_t$, and the Poisson brackets on $\mathcal{P}$ (defined by the symplectic form $-\mathbf{d}\theta_t$) are given in the following theorems. 

\begin{thm}\label{presymplectic_form}
The KMW presymplectic form $-\mathbf{d}\Theta_t$ on the space of potentials, $\mathcal{P}_o$, is given by
\begin{align}
&(\mathbf{d}\Theta_t)(X_1,X_2)=\\
&\int_{\mathcal{L}}\iota^*(\left\{\delta g_1,\delta g_2\right\}\rho_t+\delta g_1\left\{\rho_t,\delta g_2\right\}-\delta g_2\left\{\rho_t,\delta g_1\right\})\,\mu\nonumber\\
+&\int_{\mathcal{L}}\iota^*(\delta g_1\,\rho_t)\delta\mu_2-\iota^*(\delta g_2\,\rho_t)\delta\mu_1\nonumber\\
+&L_{X_1}p_{\phi,t}\bigg(\delta\phi_2-\ell_{\delta g_2}(\iota(p))\bigg)\nonumber\\
-&L_{X_2}p_{\phi,t}\bigg(\delta\phi_1-\ell_{\delta g_1}(\iota(p))\bigg),\nonumber
\end{align}
where $X_i(\iota,\phi,\mu)=(X_{\delta g_i}\circ\iota,\delta\phi_i,\delta\mu_i)$ for $i=1,2$; $\ell_f=\vartheta(X_f)-f$; and
\begin{align}
L_{X_i}p_{\phi,t}=&\int_{\mathcal{L}}\iota^*(\{\rho_t,\delta g_i\})\,\mu+\iota^*(\rho_t)\delta\mu_i.
\end{align}
\end{thm}
\begin{proof}
Using the exterior calculus identity
\begin{align}
(\mathbf{d}\Theta_t)(X_1,X_2)=L_{X_1}(\Theta_t(X_2))&-L_{X_2}(\Theta_t(X_1))\\
-\Theta_t(&[X_1,X_2]),\nonumber
\end{align}
the result can be computed directly from Eq.\,(\ref{pre_one_form}).
\end{proof} 

\begin{thm}\label{tangent}
(Weinstein)  If $(\Lambda,\tilde{\phi},\mu)\in\mathcal{P}$, there is an isomorphism
\begin{align}
T_{(\Lambda,\tilde{\phi},\mu)}\mathcal{P}\approx\bigg(C^\infty(T^*Q)\times\rm{den}(\Lambda)\bigg)/W.
\end{align}
Here $W$ is the subspace of $C^\infty(T^*Q)\times\rm{den}(\Lambda)$ consisting of elements of the form
\begin{align}
(f,L_{X_f}\mu),
\end{align}
where $f$ is an arbitrary function on $T^*Q$ that is zero along $\Lambda$ and $\iota_\Lambda:\Lambda\rightarrow T^*Q$ is the canonical inclusion.
\end{thm}
\begin{proof}
This theorem is a minor extension of Lemma 4.1 in Ref.\,\cite{Weinstein_1990}. The isomorphism is given explicitly as follows. Let $v\in T_{(\Lambda,\tilde{\phi},\mu)}\mathcal{P}$ and choose a curve in $\mathcal{P}$, $(\Lambda_t,\tilde{\phi}_t,\mu_t)$, that is tangent to $v$ at $t=0$. Because $\Lambda_t$ must satisfy the corrected Bohr-Sommerfeld condition for each $t$, there must be some Hamiltonian function $h\in C^\infty(T^*Q)$ with Hamiltonian flow $F_t$ such that $\Lambda_t=F_t(\Lambda)$. Set $f_t=F_t|\Lambda:\Lambda\rightarrow\Lambda_t$. We can therefore define the quantities
\begin{align}
\delta\phi&=\frac{\rm{d}}{\rm{d}t}\bigg|_0 f^*_t\tilde{\phi}_t-\iota_\Lambda^*\ell_{\delta g}\\
\delta\mu&=\frac{\rm{d}}{\rm{d}t}\bigg|_0 f_t^*\mu_t.
\end{align}
Note that $\delta\phi$ is constant. 

Any $h^\prime\in C^\infty(T^*Q)$ whose Hamiltonian flow $F^\prime_t$ also satisfies $\Lambda_t=F^\prime_t(\Lambda)$ must differ from $h$ by a function $c\in C^\infty(T^*Q)$ that is constant along $\Lambda$, $\delta g^\prime=\delta g+c$. Thus,
\begin{align}
\delta\phi^\prime&=\delta\phi+L_{X_c}\tilde{\phi}-\iota_\Lambda^*\ell_c\label{deltaphi}\\
\delta\mu^\prime&=\delta\mu+L_{X_c}\mu.\label{deltamu}
\end{align}
Note that the Lie derivatives in these expressions are well-defined; because $c$ is constant along $\Lambda$, $X_c$ is tangent to $\Lambda$ (See Ref.\,\cite{FoM}, Ch.\,5). Set $\delta g=h-\delta\phi$. The pair $(\delta g,\delta\mu)$ specifies a unique element of $\bigg(C^\infty(T^*Q)\times\text{den}(\Lambda)\bigg)/W$. Using Eqs.\,(\ref{deltaphi}) and (\ref{deltamu}), It is not difficult to see that this element only depends on the tangency class of the curve $(\Lambda_t,\tilde{\phi}_t,\mu_t)$. We have therefore constructed a linear map $T_{(\Lambda,\tilde{\phi},\mu)}\mathcal{P}\rightarrow\bigg(C^\infty(T^*Q)\times\text{den}(\Lambda)\bigg)/W$ given by
\begin{align}
v\mapsto(\delta g,\delta\mu)+W.
\end{align}
It is straightforward to verify that this map is an isomorphism.
\end{proof}

\begin{thm}\label{cotangent}
If $(\Lambda,\tilde{\phi},\mu)\in\mathcal{P}$, there is an isomorphism 
\begin{align}
T^*_{(\Lambda,\tilde{\phi},\mu)}\mathcal{P}\approx\bigg(\mathrm{den}(\Lambda)\times C^\infty(T^*Q)\bigg)/V.
\end{align}
Here $V$ is the subset of $\mathrm{den}(\Lambda)\times C^\infty(T^*Q)$ consisting of elements of the form
\begin{align}
(L_{X_f}\mu,-f),
\end{align}
where $f$ is an arbitrary smooth function on $T^*Q$ that is zero along $\Lambda$. The pairing between the covector $\alpha=(\kappa,\tilde{p}_\mu)+V$ and the tangent vector $v=(\delta g,\delta\mu)+W$ is given by
\begin{align}\label{pairing}
\left<\alpha,v\right>&=\int_\Lambda\,\iota_\Lambda^*\delta g\,\kappa\\
&+\int_\Lambda\iota_\Lambda^*(\left\{\tilde{p}_\mu,\delta g\right\})\,\mu+\int_\Lambda \iota_\Lambda^*\tilde{p}_\mu\,\delta\mu.\nonumber
\end{align}
\end{thm}
\begin{proof}
We will show that Eq.\,(\ref{pairing}) defines a weakly nondegenerate pairing. First observe that Eq.\,(\ref{pairing}) defines a bilinear mapping $\bigg(\mathrm{den}(\Lambda)\times C^\infty(T^*Q)\bigg)/V\times\bigg(C^\infty(T^*Q)\times\mathrm{den}(\Lambda)\bigg)/W\rightarrow\mathbb{R}$. To see this, it is only necessary to verify that the right-hand-side of Eq.\,(\ref{pairing}) is independent of the representatives of $\alpha$ and $v$. This is a simple exercise in integration-by-parts. 

To see that this bilinear map is weakly nondegenerate, fix a representative of $v$, $(\delta g,\delta \mu)$, and suppose that $\left<\alpha,v\right>=0$ for each $\alpha\in\bigg(\mathrm{den}(\Lambda)\times C^\infty(T^*Q)\bigg)/V$. Then $\left<\alpha,v\right>=0$ for all $\alpha$ of the form $\alpha=(\kappa,0)+V$. Because $\kappa$ is an arbitrary density on $\Lambda$,
\begin{align}
\iota_\Lambda^*\delta g=0,
\end{align} 
i.e. $\delta g=f$, where $f$ is some function that vanishes along $\Lambda$. This implies that, for arbitrary $\alpha$,
\begin{align}
\left<\alpha,v\right>=\int_\Lambda(\iota_\Lambda^*\tilde{p}_\mu)(\delta\mu-L_{X_f}\mu).
\end{align}
Therefore, because $\tilde{p}_\mu$ is arbitrary, 
\begin{align}
\delta\mu=L_{X_f}\mu.
\end{align}
This proves $v=(\delta g,\delta\mu)+W=(f,L_{X_f}\mu)+W=0+W$.

A very similar argument shows that if $\left<\alpha,v\right>=0$ for each $v$, then $\alpha$ must be zero.
\end{proof}

\begin{thm}\label{thm2}
The time-dependent KMW symplectic form, $-\mathbf{d}\theta_t$ on $\mathcal{P}$ is given by
\begin{align}
&\mathbf{d}\theta_t(X_1,X_2)=\\
&\int_{\Lambda}\iota_\Lambda^*(\left\{\delta g_1,\delta g_2\right\}\rho_t)\,\mu\nonumber\\
+&\int_{\Lambda}\iota_\Lambda^*(\delta g_1\left\{\rho_t,\delta g_2\right\}-\delta g_2\left\{\rho_t,\delta g_1\right\})\,\mu\nonumber\\
+&\int_{\Lambda}\iota_\Lambda^*\bigg(\delta g_1\,\rho_t\bigg)\delta\mu_2-\iota_\Lambda^*\bigg(\delta g_2\,\rho_t\bigg)\delta\mu_1,\nonumber
\end{align}
where $X_1,X_2$ are arbitrary tangent vectors at $(\Lambda,\tilde{\phi},\mu)\in\mathcal{P}$,
\begin{align}
X_i=(\delta g_i,\delta\mu_i)+W.
\end{align}
\end{thm}
\begin{proof}
By Eq.\,(\ref{def_reduced}), 
\begin{align}
\mathbf{d}\theta_t(X_1,X_2)=\mathbf{d}\Theta_t(\tilde{X}_1,\tilde{X}_2),
\end{align}
where $\tilde{X}_i$ is any vector tangent to $\mathcal{P}_o$ that satisfies $T\Pi(\tilde{X}_i)=X_i$. For each $i=1,2$, we may set
\begin{align}
\tilde{X}_i=(X_{\delta g_i}\circ\iota,\ell_{\delta g}(\iota(p)),\iota^*\delta\mu)\in T_{(\iota_o,\phi_o,\mu_o)}\mathcal{P}_o,
\end{align}
where $\Pi(\iota_o,\phi_o,\mu_o)=(\Lambda,\tilde{\phi},\mu)$. The result follows by directly computing $\mathbf{d}\Theta_t(\tilde{X}_1,\tilde{X}_2)$ using Theorem \ref{presymplectic_form}.
\end{proof}

\begin{thm}\label{thm3}
Let $F,G$ be functionals on the KMW phase space $\mathcal{P}$. For $Z=(\Lambda,\tilde{\phi},\mu)\in\mathcal{P}$, set
\begin{align}
\mathbf{d}F(Z)&=\left(\frac{\delta F}{\delta g},\frac{\delta F}{\delta\mu}\right)+V\\
\mathbf{d}G(Z)&=\left(\frac{\delta G}{\delta g},\frac{\delta G}{\delta\mu}\right)+V.
\end{align}
The time-dependent $KMW$ Poisson bracket is given by
\begin{align}\label{PB}
&\left\{\left\{F,G\right\}\right\}_t(Z)=\\
&\int_\Lambda\frac{1}{\iota_\Lambda^*\rho_t}\bigg\{\left(\iota_\Lambda^*\frac{\delta F}{\delta\mu}\right)\frac{\delta G}{\delta g}-\left(\iota_\Lambda^*\frac{\delta G}{\delta\mu}\right)\frac{\delta F}{\delta g}\bigg\}\nonumber\\
&-\int_\Lambda\iota_\Lambda^*\bigg(\rho_t\left\{\frac{1}{\rho_t}\frac{\delta F}{\delta\mu},\frac{1}{\rho_t}\frac{\delta G}{\delta\mu}\right\}\bigg)\mu.\nonumber
\end{align}
\end{thm}
\begin{proof}
The frozen-time Hamiltonian vector field, $X_{\mathcal{Q}}$, associated with a functional $\mathcal{Q}:\mathcal{P}\rightarrow\mathbb{R}$ is defined by the formula
\begin{align}
\text{i}_{X_{\mathcal{Q}}}\mathbf{d}\theta_t=-\mathbf{d}\mathcal{Q}.
\end{align}
The KMW Poisson brackets can be expressed in terms of frozen-time Hamiltonian vector fields as
\begin{align}
\left\{\left\{F,G\right\}\right\}_t=-\mathbf{d}\theta_t(X_F,X_G)=-\mathbf{d}G(X_F).
\end{align}
Using the previous theorem, it is straightforward to verify that the Hamiltonian vector field $X_F$ is given by the formula
\begin{align}\label{ham_vec}
X_F(Z)=(\delta g_F,\delta \mu_F)+W,
\end{align}
where
\begin{align}
\delta g_F&=-\frac{1}{\rho_t}\frac{\delta F}{\delta\mu}\\
\delta\mu_F&=\frac{1}{\iota_\Lambda^*\rho_t}\bigg(\frac{\delta F}{\delta g}+\iota_\Lambda^*\left\{\rho_t,\frac{1}{\rho_t}\frac{\delta F}{\delta \mu}\right\}\,\mu\bigg).
\end{align}
Using these explicit expressions for $X_F(Z)$ and Eq.\,(\ref{pairing}), the expression given above for the KMW Poisson bracket can be verified by calculating $-\mathbf{d}G(X_F)$ directly.
\end{proof}

If $\rho_t$ does not depend on time, then the KMW Poisson brackets and the KMW symplectic form are independent of time. In this special case, given a time-dependent Hamiltonian function $\mathcal{H}_t$, the associated time-dependent Hamiltonian vector field can be computed using the formula
\begin{align}\label{time_indep}
L_{X_{\mathcal{H}_t}}\mathcal{Q}=\left\{\left\{\mathcal{Q},\mathcal{H}_t\right\}\right\},
\end{align}
where $\mathcal{Q}$ is an arbitrary functional on $\mathcal{P}$.

If $\rho_t$ does depend on time, then the time-dependent Hamiltonian vector field associated with $\mathcal{H}_t$ cannot be determined using Eq.\,(\ref{time_indep}). Instead it must be determined using the time-dependent Hamilton equations mentioned earlier in the section
\begin{align}
\text{i}_{X_{\mathcal{H}_t}}\mathbf{d}\theta_t=-\mathbf{d}\mathcal{H}_t-\dot{\theta}_t.
\end{align}
In terms of the time-dependent Poisson tensor, $P_t$, defined by the time-dependent KMW Poisson brackets, $X_{\mathcal{H}_t}$ is given by
\begin{align}\label{time_dep}
X_{\mathcal{H}_t}=P_t(\mathbf{d}\mathcal{H}_t+\dot{\theta}_t).
\end{align}

Using Eq.\,(\ref{time_dep}) and the expression
\begin{align}
\dot{\theta}_t(Z)=(-(\iota_\Lambda^*\dot{\rho}_t)\mu,0)+V,
\end{align}
we obtain the following theorem
\begin{thm}
Given a time-dependent functional $\mathcal{H}_t:\mathcal{P}\rightarrow\mathbb{R}$, the associated time-dependent Hamiltonian vector field, $X_{\mathcal{H}_t}$, is given by
\begin{align}
X_{\mathcal{H}_t}(z)=(\delta g_{\mathcal{H}_t},\delta\mu_{\mathcal{H}_t})+W,
\end{align}
where
\begin{align}
\delta g_{\mathcal{H}_t}&=-\frac{1}{\rho_t}\frac{\delta \mathcal{H}_t}{\delta\mu}\\
\delta\mu_{\mathcal{H}_t}&=\frac{1}{\iota_\Lambda^*\rho_t}\bigg(\frac{\delta \mathcal{H}_t}{\delta g}-\iota_\Lambda^*\dot{\rho}_t\,\mu+\iota_\Lambda^*\left\{\rho_t,\frac{1}{\rho_t}\frac{\delta \mathcal{H}_t}{\delta \mu}\right\}\,\mu\bigg).
\end{align}
In particular, when $\mathcal{H}_t=e_t$, these expressions yield the KMW dynamical equations on $\mathcal{P}$
\begin{align}
\delta g_{e_t}&=E_t\\
\delta \mu_{e_t}&=-\iota_\Lambda^*\bigg(\frac{\dot{\rho}_t}{\rho_t}+\frac{1}{\rho_t}\left\{\rho_t,E_t\right\}\bigg)\,\mu.
\end{align}
\end{thm}

\section{Discussion\label{discussion}}
In the previous sections, we first explicitly identified the phase space for Keller-Maslov waves (KMWs), which are a generalization of eikonal waves first introduced by Keller\,\cite{Keller_1958}, and later by Maslov\,\cite{maslov}. We then proceeded to formulate the KMW dynamical equations as an infinite-dimensional Hamiltonian system. Theorem\,\ref{thm0} characterizes the KMW phase space, $\mathcal{P}$. Theorems\,\ref{thm1},\,\ref{thm2}, and\,\ref{thm3} specify the KMW variational principle, symplectic form, and Poisson brackets, respectively.

In light of these results, we clearly see how the Hamiltonian structure underlying KMW dynamics varies from wave equation to wave equation. If $D(q,p,t,U)$ is the principal symbol of the wave equation operator $\mathfrak{D}$, then the KMW Poisson bracket and energy functional only depend on $\mathfrak{D}$ \emph{via} the derived quantities $E_t$ and $\rho_t$, where $E_t(q,p)$ is defined implicitly by
\begin{align}
D(q,p,t,-E_t(q,p))=0,
\end{align} 
and $\rho_t(q,p)$ is given by
\begin{align}
\rho_t(q,p)=\frac{\partial D}{\partial U}(q,p,t,-E_t(q,p)).
\end{align}
The quantity $E_t$ is proportional to the wave frequency. The quantity $\rho_t$ is the multiplicative factor that relates wave action density (regarded as a density on the wave's Lagrangian submanifold $\Lambda$), $I_t$, to the squared modulus of the wave amplitude, $\mu_t$, i.e. $I_t=(\iota_\Lambda^*\rho_t)\,\mu_t$. 

When
\begin{align}
\mathfrak{D}=i\hbar\frac{\partial}{\partial t}-\hat{H}_t,
\end{align}
where $\hat{H}_t$ is the quantum Hamiltonian operator, $\rho=-1$. In this case the KMW symplectic form becomes very closely related to the symplectic form introduced by Weinstein in Ref.\,\cite{Weinstein_1990}. Weinstein's symplectic form is defined on an \emph{isodrast} in the space of \emph{weighted} Lagrangian submanifolds. This sort of space can be obtined from the KMW phase space by restricting attention to KMWs that have fixed total wave action, $p_\phi$, where
\begin{align}
p_\phi(\Lambda,\mu)=-\int_{\Lambda}\,\mu,
\end{align}
and then factoring out the symmetry given by global rotation of the phase factor $\tilde{\phi}$. On this reduced KMW phase space\,\cite{FoM}, the KMW symplectic form is equal to Weinstein's symplectic form.

There is already precedent for developing numerical algorithms for semiclassical wave propagation in terms the KMW phase space. References.\,\cite{Zheng_2013,Jin_2008,Liu_2006} provide recent illustrations of this fact. In particular, these works demonstrate that simulation of semiclassical wave propagation in the presence of caustics can be significantly simplified by working in terms of the abstract KMW $(\Lambda,\tilde{\phi},\mu)\in\mathcal{P}$. The results of the present study therefore suggest the relevance of the following question. Is it possible to formulate numerical integrators for semiclassical wave propagation that preserve the KMW symplectic (or Poisson) structure? One fruitful way to address this question may be to ``discretize" the variational principle given in Theorem\,\ref{thm1} in the spirit of Ref.\,\cite{Marsden_2001}.

\providecommand{\noopsort}[1]{}\providecommand{\singleletter}[1]{#1}%
%


\end{document}